\documentclass[11pt,a4paper]{article}

\usepackage[utf8]{inputenc}
\usepackage[T1]{fontenc}
\usepackage{lmodern}

\usepackage{amsmath,amssymb,amsthm}
\usepackage{bm}
\usepackage{mathtools}

\usepackage{booktabs}
\usepackage{array}
\usepackage{tabularx}
\usepackage{float}

\usepackage[margin=2.5cm]{geometry}
\usepackage{setspace}
\usepackage{parskip}

\usepackage[authoryear,round]{natbib}

\usepackage{tikz}
\usepackage{pgfplots}
\pgfplotsset{compat=1.18}
\usetikzlibrary{arrows.meta, decorations.markings, patterns, calc}

\usepackage[colorlinks=true,linkcolor=blue,citecolor=blue,urlcolor=blue]{hyperref}
\usepackage{cleveref}

\newtheorem{theorem}{Theorem}
\newtheorem{proposition}{Proposition}
\newtheorem*{proposition*}{Proposition}
\newtheorem{corollary}{Corollary}

\theoremstyle{definition}
\newtheorem{definition}{Definition}
\theoremstyle{remark}
\newtheorem*{remark}{Remark}

\newcommand{\w}{\mathbf{w}}
\newcommand{\one}{\mathbf{1}}
\newcommand{\V}{\mathbf{V}}
\newcommand{\Z}{\mathbf{Z}}
\newcommand{\Sig}{\boldsymbol{\Sigma}}
\newcommand{\muvec}{\boldsymbol{\mu}}
\newcommand{\avec}{\boldsymbol{\alpha}}

\newcommand{\SR}{\mathrm{SR}}

\newcommand{\tw}{\tau_w}
\newcommand{\tc}{\tau_c}
\newcommand{\td}{\tau_d}
\newcommand{\tk}{\tau_k}
\newcommand{\rf}{r_f}
\newcommand{\rs}{r_s}
\newcommand{\dd}{\mathrm{d}}

\title{Flow Taxes, Stock Taxes, and Portfolio Choice:\\
A Generalised Neutrality Result}
\author{Anders G Fr{\o}seth\thanks{Independent Researcher.
  E-mail: \href{mailto:indrefjorden@pm.me}{indrefjorden@pm.me}.}}
\date{\today}

\begin{document}
\maketitle

\begin{abstract}
\noindent
A proportional wealth tax---a levy on the stock of wealth---preserves
portfolio neutrality by acting as a uniform drift shift in the
Fokker--Planck equation for wealth dynamics.  We extend this result to
the full system of ownership taxes (\emph{eierkostnader}) that a
shareholder faces: a corporate tax on gross profits, a capital income
tax on the risk-free return, a dividend and capital gains tax on the
excess return, and a wealth tax on net assets.  Each tax modifies the
drift of the wealth process in a distinct way---multiplicative
rescaling, constant shift, or regime-dependent compression---while
leaving the diffusion coefficient unchanged.  We show that the combined
system preserves portfolio neutrality under three conditions: (i)~the
capital income tax rate equals the corporate tax rate, (ii)~the
shielding rate equals the risk-free rate, and (iii)~the wealth tax
assessment is uniform across assets.  When these conditions hold, the
after-tax excess return is a uniform rescaling of the pre-tax excess
return by the factor $(1 - \tc)(1 - \td)$, and the drift-shift
symmetry of the wealth-tax-only case generalises to a
\emph{drift-shift-and-rescale symmetry}.  We classify the distortions
that arise when each condition fails and show that flow-tax distortions
and stock-tax distortions are additively separable: they do not
interact.  The shielding deduction---a feature of several real-world
tax systems, including the Norwegian aksjon\ae rmodellen---emerges as
the mechanism that restores the symmetry between equity and debt
taxation within this framework.  Calibrated to the Norwegian dual income
tax, conditions (i) and~(ii) hold by institutional design; the only
binding distortion is non-uniform wealth tax assessment, which generates
portfolio tilts roughly 300 times larger than any residual flow-tax
channel.
\end{abstract}




\section{Introduction}\label{sec:intro}

The question of whether taxation distorts portfolio composition has a
long lineage.  \citet{DomarMusgrave1944} show that a proportional
income tax with full loss offset preserves the expected return per unit
of risk, so that the investor's optimal risk exposure is invariant to
the tax rate.  \citet{Stiglitz1969} extends the result to capital gains
and wealth taxation under specific conditions, and \citet{Sandmo1977}
provides a general comparative-statics treatment with many assets.

The portfolio neutrality of a proportional wealth tax specifically is
now well established.  \citet{Froeseth2026N} shows that a uniform levy
on the market value of all assets leaves portfolio weights, Sharpe
ratios, and asset prices unchanged.  \citet{Froeseth2026E} extends the result to
stochastic volatility, Epstein--Zin preferences, and general Markov
diffusions, while identifying non-uniform assessment as the principal
source of distortion in practice.  \citet{Froeseth2026S} recasts the
neutrality result as a drift-shift symmetry in the Fokker--Planck
equation for wealth dynamics: the proportional wealth tax shifts the
drift velocity uniformly, $v \to v - \tw$, without coupling to the
state, the diffusion coefficient, or the portfolio composition.

These results are derived under the simplifying assumption that the
wealth tax is the \emph{only} tax.  In practice, no shareholder faces
a wealth tax in isolation.  A Norwegian personal shareholder, for
example, faces four layers of taxation on equity ownership---the
\emph{eierkostnader} (ownership costs):
\begin{enumerate}
  \item A \emph{corporate tax} at rate $\tc$ on company profits,
        reducing the gross return available for distribution.
  \item A \emph{capital income tax} at rate $\tk$ on the risk-free
        component of the return (the portion covered by the shielding
        deduction).
  \item A \emph{dividend and capital gains tax} at rate $\td$ on the
        excess return above the shielding rate.
  \item A \emph{wealth tax} at rate $\tw$ on net assets (with
        asset-class-specific valuation discounts).
\end{enumerate}
The first three are taxes on the \emph{flow} (the return on wealth);
the fourth is a tax on the \emph{stock} (the level of wealth).  The
distinction matters for the Fokker--Planck framework: flow taxes modify
the drift through the return channel, while the stock tax modifies the
drift through a direct levy on the state variable.

This paper asks: does portfolio neutrality survive when all four taxes
are present?  The answer is yes, under three conditions that turn out to
be economically natural and, in the Norwegian case, satisfied by
institutional design.  The conditions are:
\begin{enumerate}
  \item[(C1)] The capital income tax rate equals the corporate tax rate:
              $\tk = \tc$.
  \item[(C2)] The shielding rate equals the risk-free rate:
              $\rs = \rf$.
  \item[(C3)] The wealth tax assessment is uniform across assets:
              $\alpha_i = \alpha$ for all~$i$.
\end{enumerate}

Under (C1)--(C3), the combined tax system acts on the Fokker--Planck
equation through two modifications: a uniform drift shift (from the
wealth tax) and a uniform rescaling of excess drift velocities by the
factor $(1 - \tc)(1 - \td)$ (from the flow taxes).  Neither modification
alters the \emph{relative} drifts between assets.  The drift-shift
symmetry of \citet{Froeseth2026S} generalises to a
\emph{drift-shift-and-rescale symmetry}, and all the neutrality
consequences---invariant portfolio weights, preserved Sharpe ratios,
undistorted asset prices---carry through.

When any of (C1)--(C3) fails, the symmetry breaks in a specific and
classifiable way.  We show that:
\begin{itemize}
  \item Violating (C1) or (C2) distorts the equity--debt split but
        preserves the tangency portfolio among risky assets.
  \item Violating (C3)---non-uniform wealth tax assessment---distorts
        the tangency portfolio itself, exactly as in
        \citet{Froeseth2026E}.
  \item The distortions from flow taxes and from the wealth tax are
        \emph{additively separable}: they do not interact.  The
        presence of a wealth tax neither amplifies nor dampens any
        flow-tax distortion, and vice versa.
\end{itemize}

The shielding deduction (Norwegian: \emph{skjermingsfradrag}) emerges as
the institutional mechanism that enforces condition~(C2).  By exempting
the risk-free component of the equity return from the elevated dividend
tax rate, it ensures that the boundary between the two personal tax
regimes aligns with the economic risk-free rate.  In the language of the
Fokker--Planck framework, the shielding deduction is a
\emph{symmetry-restoring device}: it prevents the two-regime tax
structure from breaking the uniform rescaling of excess drifts.

The paper is organised as follows.
\Cref{sec:framework} recapitulates the Fokker--Planck framework for
wealth dynamics and the drift-shift symmetry from
\citet{Froeseth2026S}.
\Cref{sec:flow_taxes} introduces the three flow taxes and derives how
each modifies the drift of the wealth process.
\Cref{sec:neutrality} states and proves the generalised neutrality
theorem.
\Cref{sec:breaking} classifies the symmetry-breaking channels when
each condition fails.
\Cref{sec:payment} examines how flow taxes amplify the cost of
wealth tax payment through the channels identified in Papers~1 and~3.
\Cref{sec:shielding} analyses the shielding deduction as a symmetry
restoration mechanism.
\Cref{sec:interaction} establishes the additive separability of flow
and stock tax distortions.
\Cref{sec:finance} restates the main results in the mean--variance
framework of \citet{Froeseth2026N} and \citet{Froeseth2026E}.
\Cref{sec:norway} calibrates the framework to the Norwegian dual income
tax, evaluating conditions (C1)--(C3) against the institutional design
and quantifying distortion magnitudes.
\Cref{sec:discussion} discusses the economic interpretation, the
connection to Papers~1--4, and the implications for tax policy design.
\Cref{sec:conclusion} concludes.

\section{The Fokker--Planck Framework}\label{sec:framework}

This section recapitulates the elements of the Fokker--Planck framework
developed in \citet{Froeseth2026S}---itself inspired by the
wealth-dynamics model of \citet{BouchaudMezard2000}---that are needed
for the present analysis.

\subsection{Wealth dynamics and the Fokker--Planck equation}

An investor holds a portfolio of risky assets and a risk-free asset.
Under geometric Brownian motion, the investor's wealth evolves as
\begin{equation}\label{eq:wealth_dynamics}
  \frac{\dd W}{W} = \left[\rf + \w^\top(\muvec - \rf\one)\right] \dd t
    + \w^\top\Sig\, \dd\Z \,,
\end{equation}
where $\muvec$ is the vector of expected returns, $\Sig$ is the
volatility matrix, $\V = \Sig\Sig^\top$ is the covariance matrix,
$\rf$ is the risk-free rate, and $\w$ is the vector of portfolio
weights.

For a given portfolio, the dynamics of log-wealth $x = \ln W$ are
\begin{equation}\label{eq:log_wealth}
  \dd x = v\, \dd t + \sigma_P\, \dd B \,,
\end{equation}
where $v = \rf + \w^\top(\muvec - \rf\one) - \sigma_P^2/2$ is the
drift velocity, $\sigma_P^2 = \w^\top\V\w$ is the portfolio variance,
and $B$ is a standard Brownian motion.

The Fokker--Planck equation for the density $\pi(x,t)$ of log-wealth
across an ensemble of investors is
\begin{equation}\label{eq:fp}
  \frac{\partial\pi}{\partial t}
  = -v\frac{\partial\pi}{\partial x}
    + D\frac{\partial^2\pi}{\partial x^2} \,,
\end{equation}
with diffusion coefficient $D = \sigma_P^2/2$.

\subsection{The drift-shift symmetry}\label{sec:drift_shift}

\citet{Froeseth2026S} defines the drift-shift transformation and
establishes neutrality as an invariance property.

\begin{definition}[Drift-shift transformation,
  {\citet[Definition~1]{Froeseth2026S}}]\label{def:drift_shift}
  For $\tw \geq 0$, define the map
  $\mathcal{T}_\tau: v \mapsto v - \tw$, $D \mapsto D$.
  The taxed Fokker--Planck operator is
  $\mathcal{L}_\tau = \mathcal{T}_\tau \circ \mathcal{L}_0$.
\end{definition}

\begin{proposition}[Neutrality as invariance,
  {\citet[Proposition~2]{Froeseth2026S}}]\label{prop:drift_shift}
Let assets $i$ and $j$ have drift--diffusion parameters $(v_i, D_i)$
and $(v_j, D_j)$.  Under a proportional wealth tax at rate~$\tw$:
\begin{enumerate}
  \item The difference in drift velocities is unchanged:
    $(v_i - \tw) - (v_j - \tw) = v_i - v_j$.
  \item The Sharpe-ratio-like quantity
    $(v_i - v_j)/\sqrt{D_i + D_j - 2D_{ij}}$
    is unchanged for all pairs.
  \item The optimal portfolio weights are unchanged.
\end{enumerate}
\end{proposition}

The key property is that the wealth tax shifts all drifts by the
\emph{same} constant.  Relative drifts, which determine portfolio
choice, are invariant.  In physical terms, the tax is a uniform
external field that does not couple to the internal degrees of freedom
of the system.

\subsection{Notation}\label{sec:notation}

\Cref{tab:notation} collects the symbols used throughout the paper.
The first group is inherited from Papers~1--3; the second group
contains the flow-tax parameters introduced here.

\begin{table}[ht]
\centering
\caption{Notation reference.  ``Paper'' indicates where the symbol
  first appears in the series:
  Paper~1 = \citet{Froeseth2026N},
  Paper~2 = \citet{Froeseth2026E},
  Paper~3 = \citet{Froeseth2026S},
  Paper~5 = the present paper.}\label{tab:notation}
\smallskip
\begin{tabular}{@{}lll@{}}
\toprule
Symbol & Definition & Paper \\
\midrule
\multicolumn{3}{@{}l}{\textit{Portfolio and return primitives}} \\[2pt]
$\muvec = (\mu_1,\ldots,\mu_K)^\top$ & Vector of gross expected returns & 1 \\
$\V = \Sig\Sig^\top$ & Return covariance matrix & 1 \\
$\rf$ & Risk-free rate & 1 \\
$\w = (w_1,\ldots,w_K)^\top$ & Portfolio weights on risky assets & 1 \\
$\gamma$ & Relative risk aversion (CRRA) & 1 \\
$\avec = (\alpha_1,\ldots,\alpha_K)^\top$ & Assessment fractions (wealth tax base) & 2 \\
\midrule
\multicolumn{3}{@{}l}{\textit{Tax parameters}} \\[2pt]
$\tc$ & Corporate tax rate & 5 \\
$\tk$ & Capital income tax rate & 5 \\
$\td$ & Dividend and capital gains tax rate & 5 \\
$\rs$ & Shielding rate (skjermingsrente) & 5 \\
$\tw$ & Wealth tax rate & 1 \\
\midrule
\multicolumn{3}{@{}l}{\textit{Fokker--Planck framework}} \\[2pt]
$x = \ln W$ & Log-wealth & 3 \\
$v$ & Drift velocity of log-wealth & 3 \\
$D = \sigma_P^2/2$ & Diffusion coefficient & 3 \\
$\pi(x,t)$ & Density of log-wealth across investors & 3 \\
$\mathcal{T}_\tau$ & Drift-shift transformation & 3 \\
$\mathcal{T}^{\mathrm{gen}}$ & Drift-shift-and-rescale transformation & 5 \\
\midrule
\multicolumn{3}{@{}l}{\textit{Derived quantities}} \\[2pt]
$R_i$ & After-tax return on equity asset~$i$ & 5 \\
$R_0$ & After-tax return on risk-free asset & 5 \\
$R_i^{\mathrm{ex}} = R_i - R_0$ & After-tax excess return & 5 \\
\bottomrule
\end{tabular}
\end{table}

\section{Flow Taxes as Drift Modifications}\label{sec:flow_taxes}

We now introduce three flow taxes and derive how each modifies the
drift of the wealth process.  Throughout, we consider $K$ risky assets
held in corporate form and one risk-free asset held personally
(e.g.\ bank deposits).

\subsection{Corporate tax}\label{sec:corporate}

A corporate tax at rate $\tc \in (0,1)$ is levied on company profits.
The gross return $\mu_i$ on asset~$i$ held in corporate form becomes
$\mu_i(1 - \tc)$ after corporate tax.  The risk-free rate~$\rf$,
earned on assets held personally, is not subject to corporate tax.

In log-wealth coordinates, the drift contribution of asset~$i$ changes
from $\mu_i$ to $\mu_i(1 - \tc)$.  The corporate tax is a
\emph{multiplicative} modification of the gross return:
\begin{equation}\label{eq:corporate_drift}
  v_i \;\to\; v_i^{(c)} = (1 - \tc)\mu_i - \tfrac{1}{2}\sigma_i^2 \,.
\end{equation}

\begin{remark}[Asymmetry between equity and debt]
The corporate tax applies to company profits but not to the risk-free
return on bank deposits.  This creates an asymmetry: the after-tax
return on equity is $\mu_i(1 - \tc)$, while the after-tax return on
deposits is~$\rf$.  The excess return of equity over deposits becomes
$\mu_i(1 - \tc) - \rf$, which is \emph{not} a uniform scaling of the
pre-tax excess return $\mu_i - \rf$.  The corporate tax alone therefore
distorts the equity--debt allocation.  This asymmetry is resolved when
the capital income tax is added (\Cref{sec:combined_flow}).
\end{remark}

\subsection{Capital income tax}\label{sec:cap_income}

A capital income tax at rate $\tk$ is levied on the risk-free return
component.  For bank deposits, the full return~$\rf$ is taxed:
\begin{equation}\label{eq:rf_after}
  \rf \;\to\; \rf(1 - \tk) \,.
\end{equation}
For equities, the capital income tax applies to the \emph{shielded}
portion of the return---the part up to the shielding rate~$\rs$---at
rate~$\tk$.  The shielding mechanism is described in
\Cref{sec:shielding}; for now, we work with the generic rate~$\tk$ on
the risk-free component.

\subsection{Dividend and capital gains tax}\label{sec:dividend}

A dividend and capital gains tax at rate $\td$ is levied on the
\emph{excess} of the after-corporate-tax return over the shielding
rate~$\rs$.  For an equity asset with after-corporate-tax return
$\mu_i(1 - \tc) > \rs$:
\begin{equation}\label{eq:dividend_tax}
  \text{Tax} = \td \cdot [\mu_i(1 - \tc) - \rs] \,.
\end{equation}
The dividend tax is a \emph{multiplicative} modification of the excess
return: it compresses the portion of the return above~$\rs$ by the
factor $(1 - \td)$.

\begin{remark}[Dividend--gains symmetry]
In the Norwegian system, dividends and capital gains are taxed at
the same effective rate (the nominal 22\% multiplied by the upward
adjustment factor of~1.72, giving~37.84\%).  Losses are deductible at
the same rate.  This symmetry ensures that the choice between
dividends and capital gains as a source of cash to pay the wealth tax
is not distorted by the flow taxes.  In the Fokker--Planck framework,
only the \emph{total} after-tax return matters for the drift; the
composition between dividends and realisations is irrelevant.
\end{remark}

\subsection{Combined flow-tax drift}\label{sec:combined_flow}

We now derive the drift of the wealth process under all three flow
taxes.

\subsubsection{After-tax return on equity asset~$i$}

Assume $\mu_i(1 - \tc) > \rs$.  The after-tax return on equity
asset~$i$, after corporate tax, capital income tax on the shielded
portion, and dividend tax on the excess, is:
\begin{equation}\label{eq:equity_after}
  R_i = \mu_i(1 - \tc)
    - \tk \cdot \rs
    - \td\bigl[\mu_i(1 - \tc) - \rs\bigr] \,.
\end{equation}
Collecting terms:
\begin{equation}\label{eq:equity_after_simplified}
  R_i = \mu_i(1 - \tc)(1 - \td) + \rs(\td - \tk) \,.
\end{equation}

\subsubsection{After-tax return on the risk-free asset}

\begin{equation}\label{eq:rf_total}
  R_0 = \rf(1 - \tk) \,.
\end{equation}

\subsubsection{After-tax excess return}

The after-tax excess return of equity asset~$i$ over the risk-free
asset is:
\begin{equation}\label{eq:excess_general}
  R_i^{\mathrm{ex}}
  = R_i - R_0
  = \mu_i(1 - \tc)(1 - \td) + \rs(\td - \tk)
    - \rf(1 - \tk) \,.
\end{equation}
This is the general expression.  We now examine its structure under the
conditions (C1) and (C2).

\Cref{tab:drift} summarises the drift modification introduced by each
tax layer.

\begin{table}[ht]
\centering
\caption{Drift modification by tax layer.  Each row shows how a single
  tax modifies the drift velocity $v_i$ of asset~$i$ in log-wealth
  coordinates.  The ``Character'' column classifies the modification by
  its dependence on the asset index.}\label{tab:drift}
\small
\begin{tabular}{@{}llll@{}}
\toprule
Tax & Rate & Drift modification & Character \\
\midrule
Corporate
  & $\tc$
  & $\mu_i \to \mu_i(1-\tc)$
  & Multiplicative on gross \\[4pt]
Capital income
  & $\tk$
  & $\rf \to \rf(1-\tk)$
  & Shift on risk-free drift \\[4pt]
Dividend/gains
  & $\td$
  & $(1-\td)[\mu_i(1-\tc) - \rs]$
  & Multiplicative on excess \\[4pt]
Wealth
  & $\tw$
  & $v_i \to v_i - \tw\alpha_i$
  & Uniform shift (if $\alpha_i\!=\!\alpha$) \\
\bottomrule
\end{tabular}
\end{table}

\section{The Generalised Neutrality Theorem}\label{sec:neutrality}

\subsection{Simplification under (C1) and (C2)}

\textbf{Condition (C1):} $\tk = \tc$.  The capital income tax rate
equals the corporate tax rate.

\textbf{Condition (C2):} $\rs = \rf$.  The shielding rate equals the
risk-free rate.

Under (C1) and (C2), the constant terms in \eqref{eq:excess_general}
simplify.  Substituting $\tk = \tc$ and $\rs = \rf$:
\begin{align}
  R_i^{\mathrm{ex}}
  &= \mu_i(1 - \tc)(1 - \td) + \rf(\td - \tc) - \rf(1 - \tc)
    \notag\\
  &= \mu_i(1 - \tc)(1 - \td) + \rf\td - \rf\tc - \rf + \rf\tc
    \notag\\
  &= \mu_i(1 - \tc)(1 - \td) - \rf(1 - \td)
    \notag\\
  &= (1 - \td)\bigl[\mu_i(1 - \tc) - \rf\bigr] \,.
  \label{eq:excess_simplified}
\end{align}

Now observe that $\mu_i(1 - \tc) - \rf = (1 - \tc)\mu_i - \rf$.
For the \emph{difference} between two risky assets $i$ and $j$:
\begin{equation}\label{eq:excess_diff}
  R_i^{\mathrm{ex}} - R_j^{\mathrm{ex}}
  = (1 - \td)(1 - \tc)(\mu_i - \mu_j) \,.
\end{equation}
This is a \emph{uniform scaling} of the pre-tax return differences.
The scalar factor $(1 - \tc)(1 - \td)$ is common to all asset pairs.

\subsection{Including the wealth tax}

Adding the wealth tax at rate $\tw$ with assessment fractions
$\alpha_i$ for risky asset~$i$ and $\alpha_0$ for the risk-free asset:
\begin{equation}\label{eq:excess_full}
  R_i^{\mathrm{full}}
  = (1 - \td)\bigl[\mu_i(1 - \tc) - \rf\bigr]
    - \tw(\alpha_i - \alpha_0) \,.
\end{equation}
Under \textbf{condition (C3)}, $\alpha_i = \alpha$ for all~$i$,
the wealth tax term cancels in excess returns between risky assets:
\begin{equation}\label{eq:excess_full_uniform}
  R_i^{\mathrm{full}} - R_j^{\mathrm{full}}
  = (1 - \td)(1 - \tc)(\mu_i - \mu_j) \,.
\end{equation}

\subsection{The drift-shift-and-rescale transformation}

Under (C1)--(C3), the combined tax system acts on the drift of the
wealth process through two modifications:
\begin{enumerate}
  \item A \emph{uniform rescaling} of excess drift velocities by
        $(1 - \tc)(1 - \td)$, from the flow taxes.
  \item A \emph{uniform shift} of all drifts by $-\tw\alpha$, from the
        wealth tax.
\end{enumerate}

\begin{definition}[Drift-shift-and-rescale transformation]
\label{def:generalised}
  For tax parameters $(\tc, \td, \tw, \alpha)$ satisfying (C1)--(C3),
  define the map
  \begin{equation}\label{eq:generalised_transform}
    \mathcal{T}^{\mathrm{gen}}:
    \begin{cases}
      v_i - v_0 \;\mapsto\; (1-\tc)(1-\td)(v_i - v_0) \\
      D_i \;\mapsto\; D_i
    \end{cases}
  \end{equation}
  where $v_0$ is the risk-free drift.  The overall drift level is
  shifted by $-(1-\td)\tc\rf - \tw\alpha$.
\end{definition}

This generalises \Cref{def:drift_shift}.  The drift-shift
transformation of \citet{Froeseth2026S} is the special case
$\tc = \td = 0$, $\tw > 0$.

\subsection{Main result}

\begin{theorem}[Neutrality under combined taxation]
\label{thm:main}
Let asset returns follow an It\^{o} diffusion with well-defined first
and second moments.  Let the tax system consist of a corporate tax at
rate~$\tc$, a capital income tax at rate~$\tk$, a dividend/gains tax
at rate~$\td$ with shielding rate~$\rs$, and a proportional wealth tax
at rate~$\tw$ with assessment fraction~$\alpha$.

If the following conditions hold:
\begin{enumerate}
  \item[\emph{(C1)}] $\tk = \tc$,
  \item[\emph{(C2)}] $\rs = \rf$,
  \item[\emph{(C3)}] $\alpha_i = \alpha$ for all assets~$i$,
\end{enumerate}
then the optimal portfolio weights $\w^*$ are independent of all tax
rates $(\tc, \td, \tw)$.
\end{theorem}

\begin{proof}
Under (C1) and (C2), the after-tax excess return of asset~$i$ is
$(1 - \td)(1 - \tc)(\mu_i - \rf) + c$, where $c$ is a constant
independent of~$i$ (\Cref{eq:excess_simplified}).  Under (C3), the
wealth tax contributes a further constant $-\tw\alpha$ to all excess
returns (\Cref{eq:excess_full}).

The vector of after-tax excess returns is therefore
\begin{equation}\label{eq:excess_vector}
  \mathbf{R}^{\mathrm{ex}}
  = (1 - \tc)(1 - \td)(\muvec - \rf\one) + c'\one \,,
\end{equation}
where $c' = -(1-\td)\tc\rf - \tw(\alpha - \alpha_0)$ is a scalar.

The optimal portfolio in the Markowitz problem satisfies the
first-order condition:
\begin{equation}\label{eq:foc}
  \w^* = \frac{1}{\gamma}\V^{-1}\mathbf{R}^{\mathrm{ex}} \,.
\end{equation}
Substituting \eqref{eq:excess_vector}:
\begin{equation}\label{eq:wstar}
  \w^* = \frac{(1 - \tc)(1 - \td)}{\gamma}
    \V^{-1}(\muvec - \rf\one)
    + \frac{c'}{\gamma}\V^{-1}\one \,.
\end{equation}

The first term is the untaxed tangency portfolio direction, scaled by
$(1 - \tc)(1 - \td)$.  The second term is proportional to
$\V^{-1}\one$, the global minimum variance portfolio direction.

Under CRRA preferences in continuous time, the portfolio problem takes
the form of the Merton problem.  The value function
$J(W, t) = W^{1-\gamma}f(t)/(1-\gamma)$ is homogeneous of degree
$1 - \gamma$ in wealth.  The first-order condition for the optimal
portfolio weight yields
\begin{equation}\label{eq:merton_foc}
  \w^* = \frac{1}{\gamma}\V^{-1}(\muvec^{\mathrm{after}} - R_0\one) \,,
\end{equation}
where $\muvec^{\mathrm{after}}$ is the vector of after-tax expected
returns and $R_0$ is the after-tax risk-free rate.

Substituting $\muvec^{\mathrm{after}} = (1-\tc)(1-\td)\muvec
+ [\rs(\td - \tk) - \tw\alpha]\one$ and $R_0 = \rf(1-\tk) - \tw\alpha_0$:
\begin{align}
  \w^*
  &= \frac{1}{\gamma}\V^{-1}\bigl[
    (1-\tc)(1-\td)\muvec
    + [\rs(\td-\tk) - \tw\alpha]\one
    - [\rf(1-\tk) - \tw\alpha_0]\one
  \bigr] \notag\\
  &= \frac{(1-\tc)(1-\td)}{\gamma}\V^{-1}(\muvec - \rf\one) \,,
  \label{eq:merton_wstar}
\end{align}
where the last equality uses (C1), (C2), and (C3) to cancel all
constant terms (as computed in \Cref{eq:excess_simplified} and
\Cref{eq:excess_full_uniform}).

The optimal weight $\w^*$ in \eqref{eq:merton_wstar} is proportional
to $\V^{-1}(\muvec - \rf\one)$---the same direction as the untaxed
portfolio.  The scalar $(1-\tc)(1-\td)/\gamma$ determines the total
allocation to risky assets but not the composition.  Under CRRA, the
total allocation is also invariant because the homogeneity of the value
function absorbs the scaling factor into the effective rate of time
preference (see \citet{Froeseth2026E}, Proposition~1, for the detailed
argument under stochastic volatility).

Since neither the direction nor the magnitude of $\w^*$ depends on
$\tc$, $\td$, or $\tw$, portfolio neutrality holds under the combined
tax system.
\end{proof}

\begin{remark}[Distribution-free extension]
The drift-shift-and-rescale transformation does not depend on the
Gaussian assumption.  By the same argument as in
\citet{Froeseth2026S}, Proposition~3 (distribution-free drift shift),
the result extends to general It\^o diffusions with state-dependent
drift and volatility, provided preferences are CRRA and the return
dynamics are wealth-independent.  The flow-tax rescaling factor
$(1-\tc)(1-\td)$ multiplies the excess drift regardless of the
distributional form of the noise.
\end{remark}

\Cref{fig:generalised_neutrality} illustrates the theorem in
mean--standard deviation space.  The combined tax system acts in two
stages: flow taxes rescale excess returns by $(1-\tc)(1-\td)$,
contracting the efficient frontier and the capital allocation line
toward the risk-free rate; the wealth tax then shifts the entire
opportunity set vertically by $-\tw\alpha$.  At each stage the
tangency portfolio remains at the same volatility~$\sigma^*$, and the
Sharpe ratio is preserved.

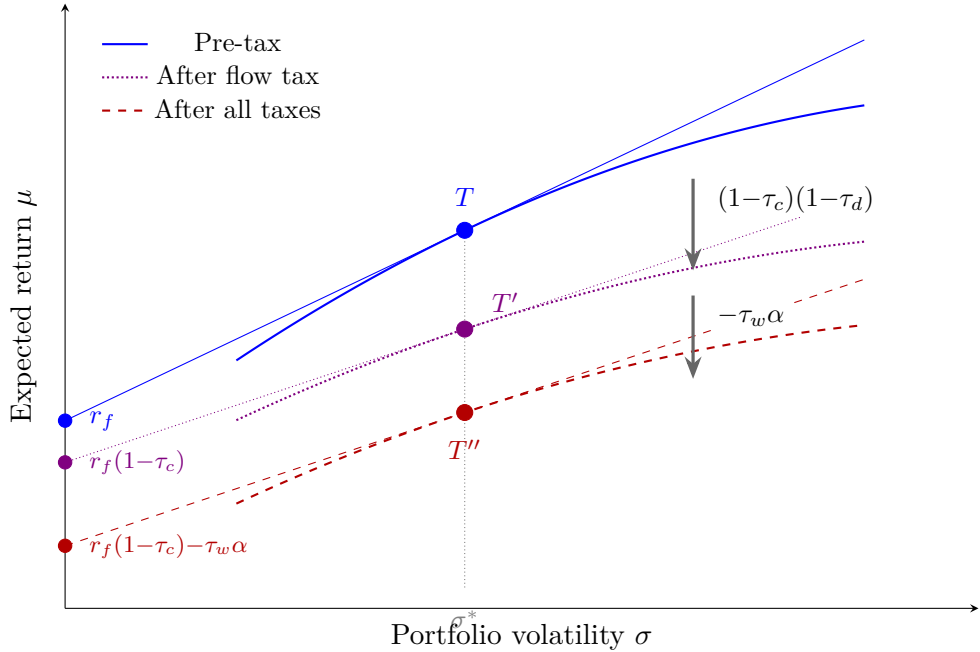
\begin{figure}[H]
\centering
\begin{tikzpicture}[>=Stealth, scale=1]
\begin{axis}[
  width=0.85\textwidth, height=0.6\textwidth,
  xlabel={Portfolio volatility $\sigma$},
  ylabel={Expected return $\mu$},
  xmin=0, xmax=0.32, ymin=-0.005, ymax=0.14,
  xtick=\empty,
  ytick=\empty,
  axis lines=left,
  every axis y label/.style={at={(axis description cs:-0.02,0.5)},
    anchor=south, rotate=90},
  legend style={at={(0.03,0.97)}, anchor=north west, font=\small,
    draw=none, fill=white, fill opacity=0.8, text opacity=1},
  clip=false
]

\addplot[blue, thick, domain=0.06:0.28, samples=80]
  ({x}, {0.02432 + 0.55*x - 0.8*x^2});
\addlegendentry{Pre-tax}

\addplot[blue!50!red, thick, densely dotted, domain=0.06:0.28, samples=80]
  ({x}, {0.019024 + 0.385*x - 0.56*x^2});
\addlegendentry{After flow tax}

\addplot[red!70!black, thick, dashed, domain=0.06:0.28, samples=80]
  ({x}, {-0.000976 + 0.385*x - 0.56*x^2});
\addlegendentry{After all taxes}

\addplot[mark=*, mark size=2.5pt, blue, only marks] coordinates {(0, 0.04)};
\node[blue, anchor=west, font=\small] at (axis cs:0.005,0.04) {$\rf$};

\addplot[mark=*, mark size=2.5pt, blue!50!red, only marks]
  coordinates {(0, 0.03)};
\node[blue!50!red, anchor=west, font=\footnotesize]
  at (axis cs:0.005,0.03) {$\rf(1{-}\tc)$};

\addplot[mark=*, mark size=2.5pt, red!70!black, only marks]
  coordinates {(0, 0.01)};
\node[red!70!black, anchor=west, font=\footnotesize]
  at (axis cs:0.005,0.01) {$\rf(1{-}\tc){-}\tw\alpha$};

\addplot[blue, thin, domain=0:0.28, samples=2]
  ({x}, {0.04 + 0.326*x});

\addplot[blue!50!red, thin, densely dotted, domain=0:0.28, samples=2]
  ({x}, {0.03 + 0.2282*x});

\addplot[red!70!black, thin, dashed, domain=0:0.28, samples=2]
  ({x}, {0.01 + 0.2282*x});

\addplot[mark=*, mark size=3pt, blue, only marks]
  coordinates {(0.14, 0.08564)};
\node[blue, anchor=south, font=\small]
  at (axis cs:0.14,0.0890) {$T$};

\addplot[mark=*, mark size=3pt, blue!50!red, only marks]
  coordinates {(0.14, 0.06195)};

\addplot[mark=*, mark size=3pt, red!70!black, only marks]
  coordinates {(0.14, 0.04195)};
\node[red!70!black, anchor=north, font=\small]
  at (axis cs:0.14,0.0380) {$T''$};

\draw[->, thick, black!60, line width=1.2pt]
  (axis cs:0.22, 0.098) -- (axis cs:0.22, 0.076);
\node[anchor=west, font=\small, fill=white, inner sep=1.5pt]
  at (axis cs:0.227, 0.093) {$(1{-}\tc)(1{-}\td)$};

\node[blue!50!red, anchor=south west, font=\small]
  at (axis cs:0.146, 0.0635) {$T'$};

\draw[->, thick, black!60, line width=1.2pt]
  (axis cs:0.22, 0.070) -- (axis cs:0.22, 0.050);
\node[anchor=west, font=\small, fill=white, inner sep=1.5pt]
  at (axis cs:0.227, 0.065) {$-\tw\alpha$};

\draw[gray, thin, densely dotted]
  (axis cs:0.14, 0) -- (axis cs:0.14, 0.08564);
\node[gray, anchor=north, font=\small] at (axis cs:0.14, -0.003)
  {$\sigma^*$};

\end{axis}
\end{tikzpicture}
\caption{Generalised neutrality in mean--standard deviation space.  Flow
taxes reduce the risk-free rate from $\rf$ to $\rf(1-\tc)$ and rescale
excess returns by $(1-\tc)(1-\td)$, contracting the efficient frontier
and capital allocation line.  The wealth tax then shifts the entire
opportunity set vertically by~$-\tw\alpha$.  At both stages the tangency
portfolio remains at the same volatility~$\sigma^*$: the combined tax
system is neutral with respect to portfolio composition.  This
generalises the vertical-translation result of
\citet[][Proposition~3]{Froeseth2026N}.}
\label{fig:generalised_neutrality}
\end{figure}

\begin{remark}[Why the constant term vanishes]
The constant $c'\one$ in \eqref{eq:excess_vector} shifts all excess
returns by the same amount.  This is the flow-tax analogue of the
wealth tax's uniform drift shift.  In the Markowitz problem, this
constant affects the total allocation to risky assets (the distance
along the capital allocation line) but not the tangency portfolio
direction.  Under CRRA in continuous time, even the total allocation
is invariant, because the constant enters the effective discount rate
and is absorbed by the homogeneity of the value function.
\end{remark}

\section{Symmetry Breaking}\label{sec:breaking}

When any of (C1)--(C3) fails, the drift-shift-and-rescale symmetry
breaks.  This section classifies the distortions by the type of
modification they introduce.

\subsection{Violation of (C1): $\tk \neq \tc$}\label{sec:break_c1}

When the capital income tax rate differs from the corporate tax rate,
the first-level taxation of equity and debt is asymmetric.  The
after-tax return on equity is scaled by $(1 - \tc)$, while the
after-tax return on deposits is scaled by $(1 - \tk)$.

The after-tax excess return of risky asset~$i$ over the risk-free
asset becomes (from \Cref{eq:excess_general}):
\begin{equation}\label{eq:break_c1}
  R_i^{\mathrm{ex}}
  = (1 - \td)\bigl[\mu_i(1 - \tc) - \rf\bigr]
    + \rf(\td - \tk) - \rf(1 - \tk)
    + \rf(1 - \td) + \rf(\tk - \tc) \,.
\end{equation}
The difference between two risky assets is still:
\begin{equation}
  R_i^{\mathrm{ex}} - R_j^{\mathrm{ex}}
  = (1 - \td)(1 - \tc)(\mu_i - \mu_j) \,,
\end{equation}
which is a uniform scaling.  \textbf{The tangency portfolio among
risky assets is preserved.}

The distortion appears in the equity--debt split.  The effective
after-tax risk-free rate is $\rf(1 - \tk)$, while equity returns are
scaled by $(1 - \tc)$.  When $\tk > \tc$, deposits are penalised
relative to equity; when $\tk < \tc$, the reverse.  This is a tilt
along the capital allocation line---a change in the total allocation
to risky assets---but not a rotation of the tangency portfolio.

\textbf{FP interpretation:} The excess drift between any two risky
assets is uniformly rescaled.  The asymmetry appears only in the drift
of equity relative to the risk-free asset, as a non-uniform constant
shift.

\subsection{Violation of (C2): $\rs \neq \rf$}\label{sec:break_c2}

When the shielding rate differs from the risk-free rate, the boundary
between the two personal tax regimes ($\tk$ and $\td$) is misaligned
with the economic risk-free rate.  The after-tax excess return contains
the term $\rs(\td - \tk) - \rf(1 - \tk)$, which under (C1) becomes
$\rs(\td - \tc) - \rf(1 - \tc)$.  When $\rs \neq \rf$, this does not
simplify to $-\rf(1 - \td)$.

The discrepancy is $(\rs - \rf)(\td - \tc)$: a constant that shifts
all equity excess returns by the same amount.  The effect is the same
in character as a violation of (C1)---a tilt along the capital
allocation line---since the constant shift does not depend on the
asset.

\textbf{The tangency portfolio is preserved.  Only the equity--debt
allocation is distorted.}

\begin{remark}[Practical relevance]
In Norway, the shielding rate is based on the arithmetic average of
three-month treasury bill (\emph{statskasseveksel}) rates, plus a
markup of 0.5 percentage points.\footnote{The markup was introduced by
the Skauge Committee \citep{NOU2003} to compensate for the asymmetry
that unused shielding is lost upon realisation.  The Torvik Committee
\citep{NOU2022} proposed replacing the current basis with ten-year
government bond yields (without the markup), arguing that the neutral
shielding rate lies between the risk-free rate and the investor's
marginal financing cost; see §8.5.7.4 of \citet{NOU2022}.}  By
construction, $\rs \geq \rf$ whenever treasury bills are a reasonable
proxy for~$\rf$, so the (C2) distortion has a known sign: it favours
equity over debt.  In either case, the magnitude is small---see
\Cref{sec:norway} for the quantitative bound---and of second order
relative to the distortion from non-uniform assessment~(C3).
\end{remark}

\subsection{Violation of (C3): $\alpha_i \neq \alpha_j$}
\label{sec:break_c3}

This is the result of \citet{Froeseth2026E}, §4.  When assessment
fractions differ across assets---a feature whose distortionary
implications for asset valuation are also analysed by
\citet{BjerksundSchjelderup2022}---the wealth tax creates an
asset-specific wedge $-\tw(\alpha_i - \alpha_0)$ in the excess return.
The portfolio distortion is:
\begin{equation}\label{eq:nonuniform_distortion}
  \Delta\w^* = -\frac{\tw}{\gamma}\V^{-1}(\avec - \alpha_0\one) \,,
\end{equation}
where $\avec = (\alpha_1, \ldots, \alpha_K)^\top$.  This \emph{does}
distort the tangency portfolio.

\textbf{FP interpretation:} The drift shift becomes asset-dependent,
$v_i \to v_i - \tw\alpha_i$.  This is the anisotropic field of
\citet{Froeseth2026S}, §7.1 (Channel~1: book-value assessment).  The
uniformity of the drift-shift transformation is broken.

\subsection{Summary}

\begin{table}[ht]
\centering
\caption{Symmetry-breaking classification.}
\label{tab:breaking}
\renewcommand{\arraystretch}{1.3}
\small
\begin{tabular}{@{}lccl@{}}
\toprule
\textbf{Condition violated}
  & \textbf{Tangency}
  & \textbf{Equity--debt}
  & \textbf{FP character} \\
\midrule
(C1): $\tk \neq \tc$
  & Preserved
  & Distorted
  & Constant shift in equity--debt drift \\
(C2): $\rs \neq \rf$
  & Preserved
  & Distorted
  & Constant shift in equity--debt drift \\
(C3): $\alpha_i \neq \alpha_j$
  & \textbf{Distorted}
  & Distorted
  & Anisotropic drift shift \\
All hold
  & Preserved
  & Preserved
  & Uniform rescaling + uniform shift \\
\bottomrule
\end{tabular}
\end{table}

\section{Flow Taxes and the Cost of Wealth Tax Payment}\label{sec:payment}

The preceding section analyses how violations of (C1)--(C3) distort
portfolio choice.  A separate question is how flow taxes interact with
the \emph{payment mechanism} of the wealth tax.

Under (C1)--(C3), the flow taxes levied on dividends or capital gains
used to pay the wealth tax do not themselves break neutrality: the
rates $\tc$ and~$\td$ are uniform across assets, so the tax cost of
generating cash is the same regardless of which asset is sold or which
firm distributes earnings.  The tangency portfolio is unaffected.

The interaction arises at a different level.
\citet{Froeseth2026N}, §9.3 establishes a payment hierarchy for the
wealth tax: investors cover the liability first from excess liquidity
and other income, then through dividend extraction from the firm, and
only as a last resort through partial sale of shares.  The wealth tax
is collected in quarterly advance instalments
(\emph{forskuddsskatt}) during the assessment year, spreading the
liquidity demand over time.  \citet{Froeseth2026S}, §7 reformulates
these payment mechanisms as distortion channels in the Fokker--Planck
framework.

Flow taxes do not create new channels, but they amplify the existing
ones by increasing the gross cost of generating cash at each stage of
the hierarchy.

\subsection{Dividend extraction}\label{sec:pay_dividends}

The most common payment mechanism for owners of private firms is
dividend extraction (\citealt{Froeseth2026N}, §9.4).  When firm
owners extract dividends to pay the wealth tax, those dividends are
taxed twice: at the corporate rate~$\tc$ and then at the personal
rate~$\td$ (above the shielding deduction).  To deliver $\tw W$ to the
investor, the firm must generate gross earnings of
\begin{equation}\label{eq:gross_dividend}
  E = \frac{\tw W}{(1 - \tc)(1 - \td)} = \frac{\tw W}{k} \,,
\end{equation}
where $k = (1 - \tc)(1 - \td)$ is the combined flow-tax pass-through
factor.  Under Norwegian rates ($\tc = 0.22$, $\td = 0.378$),
$k \approx 0.484$, so each krone of wealth tax paid requires
approximately two kroner of pre-tax earnings.

This amplifies Channel~3 of \citet{Froeseth2026S}: the drain on firm
capital becomes
\begin{equation}\label{eq:ch3_amplified}
  \frac{\dd K}{\dd t} = f(K) - \delta K - \frac{\tw W}{k} \,,
\end{equation}
rather than $f(K) - \delta K - \tw W$ in the wealth-tax-only case.
The coupling between investor wealth and firm capital is strengthened
by the factor $1/k$.  For firms facing financing frictions, this may
force additional equity issuance or reduce investment, lowering the
expected return and further reducing the drift.

\subsection{Share sales and liquidity frictions}\label{sec:pay_liquidity}

When excess liquidity and dividend capacity are insufficient to
cover the liability, the investor must sell shares.
Any realised capital gain is taxed at rate~$\td$.  To generate a net
payment of $\tw W$, the investor must sell an amount
\begin{equation}\label{eq:gross_sale}
  S = \frac{\tw W}{1 - \td \cdot g} \,,
\end{equation}
where $g$ is the fraction of the sale proceeds that constitutes a
taxable gain.  When $g > 0$, the required gross sale exceeds the tax
liability.

In the Fokker--Planck framework, this amplifies the friction cost in
Channel~2 of \citet{Froeseth2026S}:
\begin{equation}\label{eq:ch2_amplified}
  c(W, \ell) \;\to\; c\!\left(\frac{W}{1 - \td g},\, \ell\right) ,
\end{equation}
increasing the state-dependent drift modification.  The effect is
largest for investors holding appreciated, illiquid assets---precisely
those for whom the original liquidity channel is already most severe.

The quarterly instalment schedule mitigates the liquidity impact by
spreading forced sales over the assessment year rather than
concentrating them in a single transaction.  This limits the
instantaneous market impact but does not eliminate the flow-tax
amplification of the per-transaction friction cost.

\subsection{Indirect pricing effects}

The effect on prices is indirect, through the market impact channel
(Channel~5 of \citealt{Froeseth2026S}).  The flow-tax amplification
of forced sales increases the aggregate selling pressure~$F$.  The
square-root impact law---a robust empirical regularity in market
microstructure---implies that the price depression is proportional
to~$\sqrt{F}$; the low aggregate demand elasticity documented by
\citet{GabaixKoijen2022} amplifies this effect.  This is an indirect pricing
effect of the combined tax system that could not be identified in
Papers~1--2, which analysed the wealth tax in isolation: the flow taxes
roughly double the gross cash needed at each stage of the payment
hierarchy, increasing the equilibrium selling pressure and thus the
price impact.

\subsection{Migration}

The threshold wealth $W^*$ at which emigration becomes attractive
depends on the total tax cost, not just the wealth tax.  Flow taxes
reduce the after-tax return, lowering the effective wealth accumulation
rate and thereby shifting $W^*$ downward.  The absorbing-boundary
mechanism of \citet{Froeseth2026S}, §7.4 is unchanged in structure,
but the boundary moves.

\subsection{Summary}

\begin{table}[ht]
\centering
\caption{Flow-tax amplification of wealth tax payment channels.}
\label{tab:payment}
\renewcommand{\arraystretch}{1.3}
\small
\begin{tabular}{@{}llll@{}}
\toprule
\textbf{Channel}
  & \textbf{Wealth tax only}
  & \textbf{With flow taxes}
  & \textbf{Amplification} \\
\midrule
Dividends (Ch.~3)
  & $\tw W$ drain
  & $\tw W / k$ drain
  & Factor $1/k \approx 2.1$ (Norway) \\
Liquidity (Ch.~2)
  & $c(W, \ell)$
  & $c(W/(1-\td g),\, \ell)$
  & $\td g$ on realised gains \\
Market impact (Ch.~5)
  & $F(\tw)$
  & $F(\tw/(1-\td g))$
  & Larger selling pressure \\
Migration (Ch.~4)
  & $W^*(\tw)$
  & $W^*({\tw, \tc, \td})$
  & Lower threshold \\
\bottomrule
\end{tabular}
\end{table}

The dividend channel is most severely affected: under Norwegian rates,
the amplification factor $1/k \approx 2.1$ means that the effective
drain on firm capital is more than double the statutory wealth tax
liability.  The pricing effect is indirect, operating through the
market impact of amplified forced sales rather than through a direct
violation of pricing neutrality.

\section{The Shielding Mechanism}\label{sec:shielding}

The shielding deduction (\emph{skjermingsfradrag} in the Norwegian
system) is the institutional mechanism that enforces condition~(C2).
It is a form of rate-of-return allowance---an idea with deep roots in
optimal tax design
\citep{BoadwayBruce1984,Soerensen2005,Mirrlees2011}.
This section analyses its role in the neutrality framework.

\subsection{Definition and mechanics}

The shielding deduction exempts a ``normal'' return on invested capital
from the elevated dividend tax rate.  For each share, the shielding
amount is:
\begin{equation}\label{eq:shielding}
  S = \text{(cost basis)} \times \rs \,,
\end{equation}
where $\rs$ is the shielding rate (\emph{skjermingsrente}).
Institutionally, $\rs$ is set by Skattedirektoratet as the arithmetic
average of three-month treasury bill rates plus 0.5 percentage
points.\footnote{Formally, the published \emph{skjermingsrente} is
this pre-tax basis reduced by the alminnelig inntekt rate
($\rs^{\text{AT}} = \rs(1-\tk)$).  In the model, $\rs$ denotes the
pre-tax rate---the \emph{beregningsgrunnlag}---so that condition (C2)
reads $\rs = \rf$ with both rates before personal tax.}  Dividends and
capital gains up to the shielding amount are taxed at the capital
income rate~$\tk$; the excess is taxed at the dividend rate~$\td > \tk$.

\subsection{Symmetry restoration}

In the absence of shielding (i.e.\ the full equity return is taxed at
$\td$), the after-tax return on equity asset~$i$ would be:
\begin{equation}\label{eq:no_shielding}
  R_i^{\text{no shield}}
  = \mu_i(1 - \tc)(1 - \td) \,.
\end{equation}
The after-tax excess return over the risk-free asset would be:
\begin{equation}\label{eq:excess_no_shield}
  R_i^{\text{no shield}} - R_0
  = \mu_i(1 - \tc)(1 - \td) - \rf(1 - \tk) \,.
\end{equation}
For two risky assets, the difference is still
$(1 - \tc)(1 - \td)(\mu_i - \mu_j)$---a uniform scaling.  The tangency
portfolio among risky assets is preserved even without shielding.

The distortion from removing shielding appears in the equity--debt
split.  Without shielding and with $\tk = \tc$:
\begin{equation}
  R_i^{\text{no shield}} - R_0
  = (1 - \tc)(1 - \td)\mu_i - \rf(1 - \tc) \,.
\end{equation}
This is not a uniform scaling of $\mu_i - \rf$; the risk-free rate is
scaled by $(1 - \tc)$ while the equity return is scaled by
$(1 - \tc)(1 - \td)$.  The after-tax Sharpe ratio of equity relative
to deposits is distorted---equity is penalised by the additional
factor $(1 - \td)$ on the \emph{entire} return, not just the excess.

With shielding at rate $\rs = \rf$, the tax on the first~$\rf$ of the
equity return is at rate~$\tk = \tc$, matching the rate on deposits.
Only the excess $\mu_i(1 - \tc) - \rf$ is taxed at~$\td$.  The
result, derived in \Cref{eq:excess_simplified}, is:
\begin{equation}
  R_i - R_0 = (1 - \td)[\mu_i(1 - \tc) - \rf] \,,
\end{equation}
which is a uniform scaling of the after-corporate-tax excess return.

\subsection{The shielding deduction in the FP framework}

Without shielding, the equity drift is
$v_i^{\text{no shield}} = (1-\tc)(1-\td)\mu_i - \sigma_i^2/2$.
The risk-free drift is $v_0 = \rf(1-\tc) - 0$.  The excess drift
contains the term $(1-\tc)(1-\td)\mu_i - (1-\tc)\rf
= (1-\tc)[(1-\td)\mu_i - \rf]$, which mixes the scaling factors.

With shielding, the excess drift is
$(1-\td)[(1-\tc)\mu_i - \rf] = (1-\td)(1-\tc)(\mu_i - \rf) - (1-\td)\tc\rf$.
The factor $(1 - \td)(1 - \tc)$ applies \emph{uniformly} to the
pre-tax excess return $\mu_i - \rf$.  The remaining term
$-(1 - \td)\tc\rf$ is a constant independent of~$i$.

\textbf{The shielding deduction converts a mixed scaling (equity at
$(1-\tc)(1-\td)$, debt at $(1-\tc)$) into a uniform scaling
of excess returns.}  This is the symmetry restoration: the
drift-shift-and-rescale transformation becomes well-defined only when
the shielding aligns the tax regimes at the risk-free boundary.

\begin{remark}[Policy interpretation]
The shielding deduction is often justified on equity-theoretic grounds:
a ``normal'' return on capital should not face punitive taxation.  The
Fokker--Planck framework provides a complementary justification from
efficiency: the shielding deduction is the mechanism that preserves the
drift-shift-and-rescale symmetry, ensuring that the combined tax system
does not distort portfolio composition.  In this sense, the shielding
is not merely equitable---it is the condition for allocative neutrality.
\end{remark}

\section{Interaction Between Flow and Stock Taxes}
\label{sec:interaction}

Having characterised the neutrality conditions, the payment channels,
and the shielding mechanism separately, we now ask how the flow-tax
and stock-tax distortions combine.

\subsection{Additive separability}

Under (C1) and (C2), the after-tax excess return
(\Cref{eq:excess_full}) has the form:
\begin{equation}\label{eq:additive}
  R_i^{\mathrm{full}}
  = \underbrace{(1 - \td)(1 - \tc)(\mu_i - \rf)
    + \text{const.}}_{\text{flow taxes}}
  \;\;-\;\;
  \underbrace{\tw(\alpha_i - \alpha_0)}_{\text{wealth tax}} \,.
\end{equation}
The two contributions are additive.  There is no cross-term---the
flow-tax distortion (if any) and the wealth-tax distortion (from
non-uniform assessment) simply add up.

\begin{proposition}[Additive separability of distortions]
\label{prop:separability}
Under conditions (C1) and (C2), the portfolio distortion from the
combined tax system is:
\begin{equation}\label{eq:separable}
  \Delta\w^*
  = \underbrace{\Delta\w^*_{\mathrm{flow}}}_{\text{zero if (C1)--(C2) hold}}
  \;\;+\;\;
  \underbrace{\Delta\w^*_{\mathrm{stock}}}_{\text{zero if (C3) holds}} \,.
\end{equation}
The flow-tax and stock-tax distortions do not interact: the presence
of a wealth tax does not amplify or dampen any flow-tax distortion,
and vice versa.
\end{proposition}

\begin{proof}
The after-tax excess return vector is
$\mathbf{R}^{\mathrm{ex}} = (1-\tc)(1-\td)(\muvec - \rf\one) + c\one
- \tw(\avec - \alpha_0\one)$.
The optimal portfolio is
$\w^* = (1/\gamma)\V^{-1}\mathbf{R}^{\mathrm{ex}}$.  The distortion
relative to the untaxed case $\w^*_0 = (1/\gamma)\V^{-1}(\muvec - \rf\one)$
is:
\begin{align}
  \Delta\w^*
  &= \frac{1}{\gamma}\V^{-1}\bigl[
    (1-\tc)(1-\td)(\muvec - \rf\one) - (\muvec - \rf\one)
  \bigr]
  + \frac{c}{\gamma}\V^{-1}\one
  - \frac{\tw}{\gamma}\V^{-1}(\avec - \alpha_0\one) \notag\\
  &= \frac{1 - (1-\tc)(1-\td)}{\gamma}\V^{-1}(\muvec - \rf\one)
  \cdot (-1)
  + \frac{c}{\gamma}\V^{-1}\one
  - \frac{\tw}{\gamma}\V^{-1}(\avec - \alpha_0\one) \,.
\end{align}
Under CRRA in continuous time, the first two terms vanish (absorbed by
the homogeneity of the value function, as in the proof of
\Cref{thm:main}).  The remaining distortion is purely from the wealth
tax: $\Delta\w^* = -(\tw/\gamma)\V^{-1}(\avec - \alpha_0\one)$, which
is independent of~$\tc$ and~$\td$.
\end{proof}

\subsection{Drift decomposition}

In the Fokker--Planck equation, the combined tax modifies the drift as:
\begin{equation}\label{eq:drift_decomposition}
  v_i^{\mathrm{after}} - v_j^{\mathrm{after}}
  = (1 - \tc)(1 - \td)(v_i - v_j) - \tw(\alpha_i - \alpha_j) \,.
\end{equation}
The first term is a uniform rescaling of pre-tax drift differences
(flow taxes); the second is an asset-dependent shift (wealth tax
assessment).  The two terms are additive and independent.

The flow taxes rescale the drift differences without altering their
sign or relative magnitude.  The wealth tax assessment introduces an
additive distortion that is entirely determined by the assessment
fractions $\avec$ and is independent of the flow-tax rates.

\subsection{Implications for the wealth distribution}

The Pareto exponent of the stationary wealth distribution under GBM
is $\alpha = 1 + 2v/\sigma^2$ (\citealt{Froeseth2026S}, §3.2).
Under the combined tax system, the effective drift is modified by both
the flow taxes and the wealth tax.  Since the modifications are
additive, their effects on the Pareto exponent are also additive:
\begin{equation}\label{eq:pareto_combined}
  \alpha_{\mathrm{tax}} = 1 + \frac{2[v - \Delta v_{\mathrm{flow}}
    - \tw\alpha]}{\sigma^2} \,,
\end{equation}
where $\Delta v_{\mathrm{flow}}$ captures the drift reduction from the
flow taxes.  The flow taxes steepen the Pareto tail (reduce the drift)
independently of the wealth tax, and the effects compound.

\section{Results in Mean--Variance Framework}\label{sec:finance}

The preceding sections derive the main results in the Fokker--Planck
framework of \citet{Froeseth2026S}.  This section restates them in
the mean--variance and portfolio-weight formulation of
\citet{Froeseth2026N} and \citet{Froeseth2026E}.

\subsection{Opportunity set and Sharpe ratios}

\begin{corollary}[Mean--standard deviation formulation]
\label{cor:sharpe}
Under (C1)--(C3), the combined tax system contracts the
mean--standard deviation opportunity set by the factor
$(1 - \tc)(1 - \td)$ in the mean direction while leaving the
standard deviation axis unchanged.  In particular:
\begin{enumerate}
  \item The after-tax Sharpe ratio of every portfolio satisfies
    $\SR^{\mathrm{after}}(p)
    = (1 - \tc)(1 - \td)\,\SR^{\mathrm{pre}}(p)$.
    Since the scaling is uniform, the ranking of portfolios by
    Sharpe ratio is unchanged.
  \item The tangency portfolio---the portfolio that maximises the
    Sharpe ratio---is the same before and after tax.
  \item The capital allocation line contracts toward the after-tax
    risk-free rate $R_0 = \rf(1 - \tc) - \tw\alpha$ but retains
    the same slope direction.
\end{enumerate}
This generalises the orthogonality result of \citet{Froeseth2026N},
Propositions~3 and~7: the wealth tax shifts the opportunity set
vertically (drift shift); the flow taxes contract it vertically
(drift rescaling).  Both operations preserve the tangency portfolio.
\end{corollary}

\subsection{Portfolio distortions under symmetry breaking}

\begin{corollary}[Distortion classification in portfolio-weight space]
\label{cor:distortion}
The violations of (C1)--(C3) affect different components of the
optimal portfolio:
\begin{enumerate}
  \item \emph{Violations of (C1) or (C2).}  The tangency portfolio
    direction $\V^{-1}(\muvec - \rf\one)$ is preserved; only the
    total allocation to risky assets changes.  The distortion is a
    scalar adjustment to the position along the capital allocation
    line (\citealt{Froeseth2026N}, Proposition~3).
  \item \emph{Violation of (C3).}  The tangency portfolio itself is
    distorted.  The portfolio shift is
    \[
      \Delta\w^*
      = -\frac{\tw}{\gamma}\V^{-1}(\avec - \alpha_0\one) \,,
    \]
    which tilts the portfolio toward assets with lower assessment
    fractions (\citealt{Froeseth2026E}, Proposition~5).  This is
    the only violation that alters the \emph{composition} of the
    risky portfolio.
\end{enumerate}
\end{corollary}

\subsection{Reduction to the wealth-tax-only case}

\begin{corollary}[Equivalence under the Norwegian system]
\label{cor:reduction}
Under (C1)--(C3), the optimal portfolio under the combined tax system
is identical to the optimal portfolio under a wealth tax alone.  The
portfolio distortion from non-uniform assessment is exactly
\[
  \Delta\w^* = -\frac{\tw}{\gamma}\V^{-1}(\avec - \alpha_0\one)
\]
(\citealt{Froeseth2026E}, Proposition~5), regardless of the
flow-tax rates $\tc$ and~$\td$.

In the Norwegian system, where $\tk = \tc = 22\%$ and
$\rs \approx \rf$, the flow taxes contribute zero portfolio
distortion.  The entire distortion comes from non-uniform wealth
tax assessment---the result of \citet{Froeseth2026E} applies without
modification.
\end{corollary}

\section{The Norwegian Tax System}\label{sec:norway}

The general framework of \Cref{sec:flow_taxes,sec:neutrality,sec:breaking}
applies to any tax system with corporate, capital income, dividend, and
wealth tax components.  We now calibrate the framework to the Norwegian
dual income tax (\emph{aksjonærmodellen}), evaluating each of conditions
(C1)--(C3) against the institutional design and quantifying distortion
magnitudes under current parameters.

\subsection{The dual income tax and conditions (C1)--(C2)}

The Norwegian tax system taxes capital income at a flat rate of $\tk = 22\%$,
identical to the corporate tax rate $\tc = 22\%$
\citep{NOU2014,Christiansen2004}.  Condition (C1) is therefore
satisfied exactly---not approximately or as a policy aspiration, but as a
structural feature of the dual income tax design.  The Scheel Committee
\citep{NOU2014} established this symmetry, and the Torvik Committee
\citep{NOU2022} reaffirmed it as a cornerstone of the Norwegian system.

For the elevated tax on shareholder income above the normal return, the
Norwegian system uses an \emph{upward adjustment factor}
(\emph{oppjusteringsfaktor}) $f = 1.72$.  The effective dividend and
capital gains tax rate is
\begin{equation}\label{eq:oppjustering}
  \td = f \times \tk = 1.72 \times 0.22 = 0.3784 \,.
\end{equation}
This mechanism is the Norwegian implementation of the two-tier personal
tax structure: rather than imposing a separate rate on excess returns,
the system multiplies the income base by~$f$ and taxes at the ordinary
rate~$\tk$.  The factor $f$ does not appear in conditions (C1)--(C3) and
does not affect portfolio neutrality; it determines only the magnitude of
the drift rescaling through the pass-through factor
\begin{equation}\label{eq:k_norway}
  k = (1 - \tc)(1 - \td) = (1 - 0.22)(1 - 0.3784) = 0.485 \,.
\end{equation}
The gross-earnings multiplier is $1/k \approx 2.06$: for every krone of
wealth tax paid from corporate earnings, the firm must generate 2.06
kroner in pre-tax profits.

The shielding deduction (\emph{skjermingsfradrag}) sets $\rs$ on the
basis of the average three-month treasury bill rate plus a markup of
0.5 percentage points (see \Cref{sec:shielding} for details).  For
2025, the pre-tax basis is $\rs = 4.6\%$, corresponding to a published
after-tax \emph{skjermingsrente} of $\rs(1-\tk) = 3.6\%$.\footnote{%
  Source: Skatteetaten, \emph{Skjermingsrente for aksjer og
  enkeltpersonforetak}, inntekts\aa ret 2025.}  Because the markup
ensures $\rs \geq \rf$ whenever the treasury bill rate approximates the
risk-free rate, condition (C2) is satisfied or slightly over-satisfied.
The residual distortion from any discrepancy is bounded by
\[
  |(\td - \tk)(\rs - \rf)|
  = |0.3784 - 0.22| \times |\rs - \rf|
  \approx 0.16 \times |\rs - \rf| \,.
\]
With the 2025 markup of 50 basis points, this is $0.08\%$---negligible
compared to the distortion from non-uniform assessment.

\subsection{Valuation discounts and condition (C3)}\label{sec:norway_c3}

Condition (C3)---uniform wealth tax assessment---is violated.  The
Norwegian wealth tax applies asset-class-specific valuation discounts,
summarised in \Cref{tab:norway_alpha}.

\begin{table}[ht]
\centering
\caption{Wealth tax assessment fractions in Norway (2025--2026).}
\label{tab:norway_alpha}
\renewcommand{\arraystretch}{1.2}
\small
\begin{tabular}{@{}lccc@{}}
\toprule
\textbf{Asset class} & \textbf{Base} & \textbf{Discount} & \textbf{Effective $\alpha_i$} \\
\midrule
Bank deposits             & Market & 0\%  & 1.00 \\
Secondary housing         & Assessed & 0\%  & $\approx 1.00$ \\
Holiday homes (\emph{fritidsbolig}) & Assessed & 70\% & 0.30 \\
Listed shares             & Market & 20\% & 0.80 \\
Unlisted shares           & Book & 20\% & $0.80 \times B_i/M_i$ \\
Commercial real estate    & Assessed & 20\% & $\leq 0.80$ \\
Primary housing $\leq$ 10\,M NOK  & Assessed & 75\% & 0.25 \\
Primary housing $>$ 10\,M NOK (excess) & Assessed & 30\% & 0.70 \\
\bottomrule
\end{tabular}
\end{table}

The assessment fractions span a range of $\alpha_{\max} - \alpha_{\min}
= 0.75$, from bank deposits ($\alpha = 1.00$) to primary housing below
the threshold ($\alpha = 0.25$).  Listed shares receive a 20\% discount
on market value ($\alpha = 0.80$), placing them between the two
extremes.  Holiday homes (\emph{fritidsbolig}) are assessed at 30\% of
estimated market value ($\alpha = 0.30$), creating an incentive
comparable to primary housing.

The discount schedule has varied substantially over time.  Shares
carried no statutory valuation discount at all before~2017, received a
10\% discount in 2017--2018, a 45\% discount in 2021, a 25\% discount
in~2022, and the current 20\% from~2023 onward (\citealt{NOU2022},
Table~10.3).  Each change has a first-order effect on the (C3)
distortion: the tightening from 45\% to 20\% roughly tripled the
effective wealth tax on listed equity.  Housing valuations are based on
a hedonic regression model introduced in~2010, which uses transaction
data to estimate market values annually.  The 75\% discount on primary
housing has remained stable throughout and reflects a political
compromise that preserves market-based assessment while shielding
homeowners from the full tax impact.

The assessment system also applies proportional debt reduction
(\emph{gjeldsreduksjon}): debts secured against a given asset class are
reduced by the same fraction~$\beta_i = \alpha_i$ as the asset itself.
This prevents taxpayers from combining the full debt deduction with
discounted asset values, but it also means the effective discount
extends to the net-of-debt position, reinforcing the tilt toward
low-$\alpha$ assets (\citealt{Froeseth2026E}, \S4.3).

A critical distinction arises for \emph{unlisted shares}.  These are
assessed at 80\% of the company's tax-assessed net asset value
(\emph{skattemessig formuesverdi})---essentially book value---rather
than market value.  The effective assessment fraction is therefore
$\alpha_i = 0.80 \times B_i/M_i$, where $B_i$ is book value and $M_i$
is market value.  For firms with substantial intangible assets, goodwill,
or growth options, the ratio $B_i/M_i$ can be well below unity,
making the effective assessment fraction much lower than~0.80.  This is
the book-value taxation channel identified in
\citet{Froeseth2026N}, \S9.2: the wealth tax base becomes a
deterministic liability rather than a proportional claim on the
stochastic market value, breaking the uniformity of the drift shift.

The resulting portfolio distortion follows from the formula of
\citet{Froeseth2026E}, Proposition~5:
\[
  \Delta\w^* = -\frac{\tw}{\gamma}\V^{-1}(\avec - \alpha_0\one) \,.
\]
Since $\alpha_{\text{listed}} < \alpha_{\text{deposits}}$, the wealth
tax tilts portfolios toward shares and away from deposits.  Since
$\alpha_{\text{unlisted}}$ can be well below $\alpha_{\text{listed}}$
(depending on the book-to-market ratio), the tilt toward unlisted
shares can be substantially larger than toward listed shares.  Since
$\alpha_{\text{housing}} < \alpha_{\text{listed}}$, there is a further
tilt toward primary housing.  The combined effect is an asset allocation
that favours owner-occupied housing and unlisted equity in asset-light
firms over bank deposits and listed shares.

\subsection{Distortion magnitudes under current parameters}
\label{sec:norway_magnitudes}

\Cref{tab:norway_distortions} compares the distortion magnitudes from
each condition.

\begin{table}[ht]
\centering
\caption{Distortion magnitudes under the Norwegian tax system.}
\label{tab:norway_distortions}
\renewcommand{\arraystretch}{1.2}
\small
\begin{tabular}{@{}llll@{}}
\toprule
\textbf{Condition} & \textbf{Parameter values}
  & \textbf{Distortion bound}
  & \textbf{Order} \\
\midrule
(C1): $\tk \neq \tc$
  & $\tk = \tc = 22\%$
  & $|\tk - \tc| \cdot \rf = 0$
  & Exactly zero \\
(C2): $\rs \neq \rf$
  & $|\rs - \rf| \leq 50\;\text{bp}$
  & $0.16 \times 0.005 = 0.08\%$
  & Negligible \\
(C3): $\alpha_i \neq \alpha_j$
  & $\Delta\alpha \leq 0.75$
  & $\tw \cdot \Delta\alpha / \gamma$
  & $0.75\%/\gamma$ \\
\bottomrule
\end{tabular}
\end{table}

For a risk aversion parameter $\gamma = 3$, the (C3) distortion is of
order $0.25\%$ per unit of inverse covariance---roughly 300 times
larger than the upper bound on the (C2) channel, and the only nonzero
channel among the three.

To put this in portfolio terms, consider the two-asset case of listed
shares versus bank deposits.  With $\tw = 1.0\%$,
$\alpha_{\text{listed}} = 0.80$, $\alpha_{\text{deposits}} = 1.00$,
and $\gamma = 3$:
\[
  \Delta w^* = \frac{0.01 \times 0.20}{3 \times \sigma^2}
  = \frac{0.00067}{\sigma^2} \,.
\]
For an annual volatility of $\sigma = 0.20$ (typical for the Oslo
Stock Exchange), this gives $\Delta w^* \approx 1.7$ percentage points
toward listed shares.

For unlisted shares with a book-to-market ratio of $B/M = 0.5$
(representative of asset-light or high-growth firms), the effective
assessment is $\alpha_{\text{unlisted}} = 0.80 \times 0.5 = 0.40$,
and the tilt versus deposits is roughly four times larger: $\Delta w^*
\approx 5$ percentage points.  This illustrates why the book-value
assessment of unlisted equity, identified in \citet{Froeseth2026N},
\S9.2, creates a stronger portfolio incentive than the statutory 20\%
discount alone suggests.

For primary housing versus deposits, with $\Delta\alpha = 0.75$:
\[
  \Delta w^* = \frac{0.01 \times 0.75}{3 \times \sigma^2_h} \,,
\]
which for a housing volatility of $\sigma_h = 0.10$ yields
$\Delta w^* \approx 25$ percentage points.  The valuation discount on
primary housing remains the dominant source of portfolio distortion in
the Norwegian system (\citealt{Froeseth2026E}, \S4.3), but the
book-value channel for unlisted shares is a significant secondary
source whose magnitude depends on the firm's asset composition.

\subsection{Beyond the proportional framework}\label{sec:norway_beyond}

Several features of the Norwegian system fall outside the proportional
framework of \Cref{sec:neutrality}.

\textbf{Shielding accumulation.}  Unused skjermingsfradrag carries
forward and compounds at the shielding rate~$\rs$.  An investor holding
shares with cost basis~$B$ for $n$ years without fully utilising the
deduction accumulates approximately $B[(1 + \rs)^n - 1]$ in unused
shielding.  This accumulated shielding is \emph{lost upon realisation}:
it cannot be transferred to other shares or set against other income.
The loss-upon-realisation creates a lock-in effect that makes the
effective tax rate path-dependent and is not captured by the static
framework \citep{Auerbach1991}.  Quantifying this channel requires
a dynamic model with heterogeneous holding periods.

\textbf{Progressive wealth tax.}  The Norwegian wealth tax has a
threshold of 1.9\,M~NOK (single, 2026) and a two-bracket rate
structure: 1.0\% on wealth between 1.9\,M and 21.5\,M~NOK, and 1.1\%
above 21.5\,M~NOK.  This introduces a wealth-dependent effective rate
$\tw(W)$, creating a state-dependent drift modification that breaks
the uniformity of the drift shift.  In the Fokker--Planck framework,
this is exactly the confining potential analysed in
\citet{Froeseth2026R}: the progressive structure generates a restoring
force that compresses the wealth distribution.  The threshold also
creates a tax shield in the sense of \citet{Froeseth2026E},
Proposition~6, which increases risk-taking for investors near the
exemption boundary.

\textbf{Retention and deferral.}  If a firm retains earnings, the
dividend tax is deferred until distribution.  This creates a timing
option whose value depends on payout policy, future tax rates, and the
investor's discount rate.  The effective tax rate becomes
path-dependent, introducing a channel that favours retention and
complicates the uniform-rescaling result.  The participation exemption
amplifies this deferral channel.

\textbf{The participation exemption (\emph{fritaksmetoden}).}
Corporate shareholders---aksjeselskaper, foreninger, and
stiftelser---are exempt from tax on dividends and capital gains under
the \emph{fritaksmetoden}, introduced as part of the 2006 tax reform
to prevent chain taxation (\emph{kjedebeskatning}) of corporate
profits through ownership layers (\citealt{NOU2022}).  A partial
correction was added in 2009: 3\% of exempt income is included in
ordinary income, yielding an effective tax rate of $0.03 \times 0.22
= 0.66\%$ on inter-company distributions.  Dividends within a tax
group (ownership above 90\%) are fully exempt.  Since most substantial
Norwegian investors hold equities through unlisted holding companies,
the fritaksmetoden is the dominant ownership structure for taxable
wealth.  A consequence, discussed in \citet{Froeseth2026N}, \S9.2, is
that the wealth tax base for these investors is assessed at
beginning-of-period book value (1~January of the income year) rather
than end-of-period market value---a predetermined quantity that
simplifies the pricing equation and is empirically important for
calibration.

For the framework of this paper, the fritaksmetoden changes the
effective pass-through factor.  A personal investor holding shares
directly faces $k = (1 - \tc)(1 - \td) \approx 0.485$.  The same
investor routing ownership through a holding company faces
$k_{\text{hold}} \approx (1 - \tc)(1 - 0.0066) \approx 0.775$ on
reinvested corporate earnings, because the dividend tax is replaced by
the 0.66\% effective rate within the corporate sector.  The personal
dividend tax is deferred until the investor extracts funds from the
holding company.  The result is a two-tier system: the tangency
portfolio is preserved (the rescaling is still uniform across assets),
but the \emph{magnitude} of the rescaling differs between direct and
indirect ownership, creating an incentive to interpose holding
structures.  \citet{BjerksundSchjelderup2021a} analyse the neutrality
of this two-tier system and show that the investor is indifferent
between direct and indirect ownership only when the shielding rate
equals the investor's borrowing rate---a condition that generally does
not hold, since the institutional $\rs$ is based on the three-month
treasury bill rate plus a 0.5 percentage point markup (see
\Cref{sec:shielding}), which lies below typical borrowing rates.  \citet{BjerksundSchjelderup2021b} document the
resulting lock-in: retained capital in Norwegian holding structures
grew from approximately 500\,M to 2,600\,M~NOK between 1999 and~2019,
and they use a Black--Scholes framework to value the embedded option in
the shielding system, finding forward equity tax credits of 0.36--1.34
NOK per krone invested depending on holding period and firm maturity.
The combination of the participation exemption and the shareholder model
favours established investors with holding structures---who can defer
personal dividend tax indefinitely---over those holding shares directly.

The quantitative significance is documented by
\citet{BjerksundHoplandSchjelderup2024}, who compute effective tax
rates for a tax-minimising investor holding unlisted shares through a
holding structure.  Using the average effective corporate tax rate of
13.1\% (2004--2018) and an effective valuation discount of 65\% on
unlisted shares---reflecting the book-value assessment channel of
\Cref{sec:norway_c3}, not just the statutory verdsettelsesrabatt---they
find an overall effective tax rate of approximately 14.4\%.  This is
well below the 21.4\% for a passive investor in listed shares and the
27.3\% average for a single worker.  The fritaksmetoden accounts for
the bulk of the gap: by eliminating the personal-level dividend tax on
retained earnings, it reduces the effective burden from the combined
tax system by roughly one-third for investors who can defer
distributions indefinitely.

\textbf{Wealth tax deferral.}  From 2026, Norway introduces a
deferral scheme (\emph{utsettelsesordning}) allowing taxpayers with
taxable wealth above the threshold to defer wealth tax payments for up
to three years, with interest accruing at the market rate.  In the
framework of \Cref{sec:payment}, this converts the immediate
liquidity drain into a deferred obligation, reducing the payment-channel
amplification for constrained investors.  The deferral does not change
the present-value tax burden---interest ensures NPV-equivalence---but it
alters the timing of the forced asset liquidation, which may affect
portfolio dynamics in practice.

\section{Discussion}\label{sec:discussion}

\subsection{Connection to the paper series}

The results of this paper extend the neutrality framework of Papers~1--4
in a natural way.  \Cref{tab:connections} summarises the connections.

\begin{table}[ht]
\centering
\caption{Connections to the paper series.}
\label{tab:connections}
\renewcommand{\arraystretch}{1.3}
\small
\begin{tabular}{@{}lll@{}}
\toprule
\textbf{Prior result} & \textbf{Extension here} \\
\midrule
Paper~1, Prop.~2: Portfolio invariance (GBM)
  & Extends to combined flow + stock taxes \\
Paper~1, Prop.~3: Orthogonality
  & Combined taxes = uniform rescaling + shift \\
Paper~2, Prop.~1: Stochastic vol.\ neutrality
  & Same mechanism: CRRA absorbs flow taxes \\
Paper~2, Prop.~5: Non-uniform assessment
  & Remains the binding constraint \\
Paper~3, Def.~1: Drift-shift transformation
  & Generalises to drift-shift-and-rescale \\
Paper~3, Prop.~2: Neutrality as invariance
  & Relative drifts scaled by $(1-\tc)(1-\td)$ \\
Paper~4, Prop.~1: Gini preservation
  & Extends: uniform rescaling preserves Gini \\
Paper~7, Thm.~1: Spectral invariance
  & Pass-through factor $k$ is isotropic perturbation \\
\bottomrule
\end{tabular}
\end{table}

The connection to Paper~7 deserves elaboration.  \citet{Froeseth2026X}
develops spectral portfolio theory by identifying neural network weight
matrices as portfolio allocation matrices and proving that any
\emph{isotropic} perturbation to the portfolio objective preserves the
singular-value distribution up to scale and shift (Theorem~1).  The
pass-through factor $k = (1 - \tc)(1 - \td)$ of the present paper is
exactly such an isotropic perturbation: it rescales all excess returns
uniformly, treating every asset symmetrically.  The
drift-shift-and-rescale symmetry of \Cref{thm:main} is therefore the
scalar projection of Paper~7's spectral invariance.  Conversely, the
symmetry-breaking classification of \Cref{sec:breaking} maps onto
Paper~7's isotropic/anisotropic distinction: conditions (C1) and~(C2)
ensure that the flow-tax rescaling is isotropic, while condition~(C3)
is the non-anisotropy requirement for the wealth tax---the same
condition that prevents spectral distortion in the matrix-valued
framework.

\subsection{Equilibrium prices under inelastic markets}

\citet{Froeseth2026E}, §5 analyses how the wealth tax affects
equilibrium prices when asset supply is inelastic, following the
inelastic markets hypothesis of \citet{GabaixKoijen2022}.  The
conclusion is that a uniform wealth tax depresses the price level of
risky assets but does not change relative prices, because the drift
shift is the same for all assets.

The same logic extends to the combined tax system.  Under
(C1)--(C3), the flow taxes rescale all excess returns by the uniform
factor $k = (1 - \tc)(1 - \td)$.  This reduces the after-tax excess
return on every risky asset by the same proportion, lowering aggregate
demand for risky assets.  Under inelastic supply, the equilibrium price
level falls.  But because the rescaling is uniform, relative prices are
unchanged---the portfolio composition is unaffected.

The combined effect on the price level is larger than under the wealth
tax alone: the wealth tax shifts the drift by $-\tw\alpha$, and the
flow taxes compress the remaining excess drift by the factor~$k$.
Both effects reduce demand, and they compound.  However, neither
modifies relative prices, so pricing neutrality in the
cross-sectional sense is preserved.  This extends the result of
\citet{Froeseth2026E}, §5 to the combined tax system.

\subsection{The role of the shielding deduction}

The shielding deduction has received attention in the tax policy
literature primarily as an equity device: it ensures that a ``normal''
return on capital is not taxed at the elevated shareholder rate
(\citealt{Soerensen2005}).  The present analysis adds a complementary
perspective from allocative efficiency.

In the Fokker--Planck framework, the shielding deduction is the
mechanism that aligns the boundary between the two personal tax regimes
with the risk-free rate, converting a potentially distortionary
two-regime tax structure into a uniform rescaling of excess returns.
Without it, the dividend tax would apply to the \emph{entire}
after-corporate-tax return, creating an asymmetry between equity and
debt that would distort the equity--debt allocation (though not the
tangency portfolio among equities).  With it, the only remaining source
of portfolio distortion is non-uniform wealth tax assessment.

This provides a formal efficiency argument for the shielding
deduction that is independent of the equity argument.

\subsection{Policy implications}

The Norwegian calibration of \Cref{sec:norway} yields a clear policy
ranking: equalising assessment fractions across asset classes would do
far more for allocative neutrality than any reform of the flow-tax rates
or the shielding mechanism.  The (C3) channel generates distortions
roughly 300 times larger than the residual (C2) channel and is the only
nonzero source of portfolio distortion in the Norwegian system.  This
echoes the empirical findings of \citet{Jakobsen2020} and
\citet{Brulhart2022}, who document substantial portfolio responses to
wealth tax incentives in Scandinavian and Swiss data.

\subsection{Limitations}

Beyond the Norwegian-specific qualifications discussed in
\Cref{sec:norway_beyond} (shielding accumulation, progressive brackets,
retention and deferral), one general limitation deserves mention.

\textbf{Discrete time.}  In discrete time, the scaling factor
$(1 - \tc)(1 - \td)$ may affect the total risky allocation under
non-CRRA preferences, even though it preserves the tangency portfolio.
The magnitude of this effect for realistic parameters is an empirical
question.

\section{Conclusion}\label{sec:conclusion}

The combined system of ownership taxes---corporate tax, capital income
tax, dividend tax, and wealth tax---preserves portfolio neutrality under
three conditions: equal rates on capital income and corporate profit,
alignment of the shielding rate with the risk-free rate, and uniform
wealth tax assessment.

The drift-shift symmetry of \citet{Froeseth2026S} generalises to a
drift-shift-and-rescale symmetry: the flow taxes uniformly rescale the
excess drift velocities by $(1 - \tc)(1 - \td)$, while the wealth tax
uniformly shifts all drifts.  Neither modification alters relative
drifts between assets.  The distortions from flow taxes and from the
wealth tax are additively separable and do not interact.

The shielding deduction emerges as a symmetry-restoring mechanism in
the Fokker--Planck framework.  By aligning the boundary between the
two personal tax regimes with the risk-free rate, it converts a
potentially distortionary two-regime structure into a uniform rescaling
of excess returns.  This provides a formal efficiency argument for the
shielding deduction that complements the standard equity justification.

Calibrated to the Norwegian dual income tax, the framework confirms
that flow-tax neutrality holds by design: $\tk = \tc = 22\%$ and the
shielding rate tracks the risk-free rate.  The pass-through factor
$k = (1 - \tc)(1 - \td) \approx 0.485$ compresses excess drifts but
preserves their ranking.  The dominant---and effectively the
only---source of portfolio distortion is non-uniform wealth tax
assessment, with valuation discounts ranging from 0\% (bank deposits) to
75\% (primary housing).  The resulting portfolio tilts are roughly 300
times larger than the residual distortion from any shielding-rate
misalignment.  Equalising assessment fractions would do more for
allocative neutrality than any reform of the flow-tax structure.

\subsection*{Acknowledgements}
The author acknowledges the use of Claude (Anthropic) for assistance with
literature review, \LaTeX{} typesetting, mathematical exposition, and
editorial refinement, and Lemma (Axiomatic AI) for review and proof
checking. All substantive arguments, economic reasoning, and conclusions
are the author's own.

\bibliographystyle{plainnat}

\end{document}